\documentclass[aps,prl,onecolumn,notitlepage,superscriptaddress]{revtex4-1}
\usepackage{epsfig,amssymb,amsmath,mathrsfs,amsfonts,amsthm,graphicx,float,color}
\usepackage[active]{srcltx}
\usepackage{subfigure}

 \def\map#1{\mathcal #1}
\def\d{\operatorname{d}}\def\<{\langle}\def\>{\rangle}

\def\Tr{\operatorname{Tr}}\def\:{\hbox{\bf
    :}}

\def\grp#1{\mathsf{#1}}

\DeclareMathOperator\erf{erf}

\def\set#1{\mathsf{#1}}

\newtheorem{theo}{{Theorem}}
\newtheorem{lem}{{Lemma}}
\newtheorem{prop}{{Proposition}}

\begin{document}


\title{Quantum Stopwatch: How To Store Time in a Quantum Memory}
\author{Yuxiang Yang}
\affiliation{Department of Computer Science, The University of Hong Kong, Pokfulam Road, Hong Kong}
\affiliation{HKU Shenzhen Institute of Research and Innovation Yuexing 2nd Rd Nanshan, Shenzhen 518057, China}
\author{Giulio Chiribella} 
\affiliation{Department of Computer Science, The University of Oxford,  Parks Road, Oxford, UK}
\affiliation{Department of Computer Science, The University of Hong Kong, Pokfulam Road, Hong Kong}
\affiliation{Canadian Institute for Advanced Research, CIFAR Program in Quantum Information Science, Toronto, ON M5G 1Z8}

\author{Masahito Hayashi}
\affiliation{Graduate School of Mathematics, Nagoya University, Nagoya, Japan}
\affiliation{Centre for Quantum Technologies, National University of Singapore, Singapore}


\begin{abstract}

Quantum mechanics imposes a fundamental tradeoff between the accuracy of time measurements and the size of the systems used as clocks.  
    When the measurements of different time intervals are combined,   the  errors due to the finite clock size accumulate, resulting in  an overall  inaccuracy that grows with   the complexity of the setup.      
 Here we   introduce a method  that in principle eludes the accumulation of errors by coherently transferring   information from  a quantum clock  to a quantum memory of the smallest possible size.
 Our method could be used to  measure the total duration of a sequence of events with enhanced accuracy, and to reduce the  amount of quantum communication needed to stabilize  clocks in a quantum network.
\end{abstract}
\maketitle

\section{Introduction}

Accurate time measurements are important  in a variety  of  applications, including  GPS systems \cite{klepczynski1996gps}, frequency standards \cite{fs}, and astronomical observations \cite{astronomy1,astronomy2}.  But the accuracy of time measurements is not just a  technological issue.    
At the most fundamental level, every clock is subject to an unavoidable  quantum limit, which cannot be overcome even with the most advanced technology.  The limit has its roots in Heisenberg's uncertainty principle, which  implies  fundamental bounds on the accuracy  of time measurements    \cite{mandelstam1945uncertainty,helstrom,holevo}.    
For  an individual time measurement, the ultimate quantum limit  can be attained by initializing the clock in a  suitably engineered  superposition of energy levels \cite{buzek-temporal,mixed-clock1,mixed-clock2}.  However, the situation is different when multiple time measurements are performed on the same clock (e.~g.~in order to measure the total duration of a sequence of events) or on different clocks (e.~g.~in   GPS technology).  In these scenarios the errors accumulate, resulting in an inaccuracy that grows linearly with the  number of measurements. 
   To address this problem,  one may try to minimise  the number of  measurements: instead of measuring individual clocks, one could store their state into the  memory of a quantum computer, process all  the data  coherently, and finally read out the  desired information with a  single  measurement.  
   However,    quantum memories  are notoriously   expensive and hard to scale up   \cite{memory}.    
This leads to the  fundamental  question:  how much memory is required to record time at the quantum level?

 \begin{figure}[t]
 \centering
\includegraphics[width=0.5\linewidth]{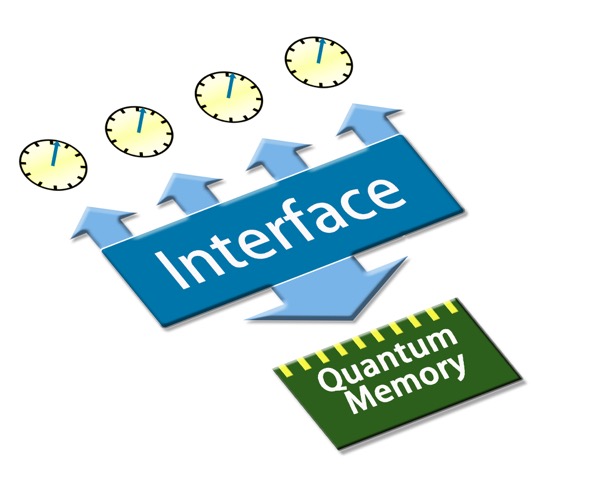} \caption{{\bf Working principle of the  quantum stopwatch.}
The quantum stopwatch   coherently transfers time information from a quantum clock, consisting of many identical particles, to a quantum memory of minimum size. 
} \label{fig:clock} 
\end{figure}

Here we derive the ultimate quantum limit on the amount of memory needed to record time   with a prescribed  accuracy.  The limit is based on a Heisenberg-type bound, expressing the tradeoff between the accuracy in the read-out of a given parameter and the size of the system in which the parameter is encoded. 
 We show that the bound is tight, by constructing a  protocol that faithfully transfers information from the system  to a quantum memory of minimal size. The protocol, which we call {\em quantum stopwatch},  freezes the time evolution of a clock by storing its state  into the state of the memory, as in   Figure \ref{fig:clock}.   The quantum stopwatch protocol works with clocks made of many identical and independently prepared particles,  a common setting  when the clocks are identical atoms or ions \cite{qubit-number}. The use of identical particles can also be thought as a simple repetition code for transmitting time information.   
  Since our protocol uses identically prepared particles,  the optimal scaling with the memory size is robust to depolarization of the clocks and to particle loss. 
   

   Storing time coherently  into a quantum memory  is a useful primitive for many applications.  As an illustration,  we construct a quantum-enhanced protocol to measure  the total duration of  a  sequence of events.  The same protocol can be used  to establish a shared frequency standard  among the nodes of a network, and to  generate    quantum states with Heisenberg-limited sensitivity to time evolution and to  phase shifts.


\section{The size-accuracy tradeoff}    
Suppose that  a  parameter $T$ is encoded in the state of a  quantum system, say $\rho_T$.    The system can be either a quantum clock, where $T$ is the time elapsed since the beginning of the evolution, or a quantum memory, where the dependence of $\rho_T$ on $T$ can be completely arbitrary.   In general, the parameter  $T$ does not have to be time: it can be phase, frequency, or any other real parameter.  

When needed, one can extract information about the parameter $T$ by measuring the system. 
The question is how accurate the measurement can be. The  inaccuracy of a given measurement can be quantified by the size of the smallest interval, centred around the true value,  in which the measurement outcome  falls with a prescribed probability---for example, $P= 99\%$. Explicitly, the inaccuracy has the expression
\begin{align}\label{A1}
\delta (P,T):= \inf    \left\{ { \delta }   \, \Big|  ~  P(\delta,T)      \ge  P   \right\}   \, ,
\end{align}
where   $P(\delta, T)$ is the probability that  the measurement outcome $\widehat T$ belongs to an interval of size $\delta$ centred around the true value $T$.

Note that  the inaccuracy can generally depend on the true value $T$, which is unknown to the experimenter. 
The dependence can be removed by fixing  a fiducial interval $[T_{\min},  T_{\max}]$.  
For example, the fiducial interval could be the inversion region where the parameter $T$ is in one-to-one correspondence with the state of the system \cite{kohlhaas2015phase}.    We denote by $\delta (P)$ the worst-case value of the inaccuracy within the fiducial interval.     In the Bayesian approach,  $\delta (P)$ provides  a lower bound on the probability that the true value falls within  an interval of size  $\delta (P)$ around the measured value  $\widehat T$: such probability is guaranteed to be at least $P$ for every prior distribution on $T$ and for all values of $\widehat T$  except at most a zero-probability set.    Other  properties of the inaccuracy, used later in the paper, are presented in Methods.

We now derive a   fundamental  lower bound on  the inaccuracy, expressed in terms of the size of the quantum system used to encode  the parameter  $T$.   Let us denote by  $D$ the dimension of the smallest  subspace containing the eigenvectors of the states $\{\rho_T  ~|~  T_{\min}\le T\le   T_{\max}\}$.  Physically, $D$ can be regarded as the {\em effective dimension} of the system used to encode the parameter $T$.  In terms of the effective dimension, the inaccuracy satisfies the bound
\begin{align}\label{HL-memory}
 {\delta  (P)}  \,   \ge \frac {P \,  \Delta  T }{D+1}\, ,   \qquad \Delta T  :=  T_{\max} -  T_{\min}\, ,
\end{align}
valid for arbitrary encodings of the parameter $T$ and for arbitrary quantum measurements.    We call Eq. (\ref{HL-memory}) the {\em size-accuracy bound}.    
   
The  size-accuracy bound   follows from dividing the fiducial interval  $\Delta T$ into   $ N=  \lfloor \Delta T/\delta (P) \rfloor$ disjoint intervals of size $\delta (P)$.  One can then encode the midpoint value of the $i$-th interval into the state   $\rho_{T_i}$.  In this way, one obtains $N$ quantum states, which can be distinguished with probability of success at least $P$ (one has just to estimate $T$ and to declare the state $\rho_{T_i}$ if the estimate of $T$ falls in the $i$-th interval).       On the other hand,  the $N$ states are contained in an $D$-dimensional subspace,  and therefore the probability of success is upper bounded by $ D/N$ \cite{yuen-1975}, leading to the bound $P\le   D/N$ and, in turn, to  Eq.  (\ref{HL-memory}).

The size-accuracy bound  (\ref{HL-memory}) captures in a unified way  the Heisenberg scaling  of quantum clocks and the ultimate limits on the memory needed to store the parameter  $T$.    Let us first see how it implies the Heisenberg scaling of quantum clocks.  Consider  a clock made of $n$ identical non-interacting particles, each evolving with the same periodic evolution  $U_T   =   e^{-i  T  H/\hbar}$, where $H$ is the single-particle Hamiltonian.  If the particles are initialized in the state $|\Psi\>$, then  the quantum state at time $T$ is    $|\Psi_T\>  =    U_T^{\otimes n} \,  |\Psi\>$.    Now, in order for the time evolution to be periodic, the eigenvalues of $H$ must be  integer multiples of a given energy.      This implies that   the number of distinct eigenvalues of the $n$-particle   Hamiltonian grows linearly with $n$, and, therefore,  all the states  $|\Psi_T\>$ are contained in a subspace of  dimension proportional to $ n$. Hence, one obtains  the bound $  \delta (  P ) \ge {c}/ n$ for some suitable constant $c>0$. 
Note that this relation also holds for mixed states, because mixing can only increase the inaccuracy  (see Methods).  
Moreover, the relation    $  \delta (  P ) \ge {c}/ n$  holds for arbitrary measurements. 
   
   The bound $\delta (P)\ge c/n$ implies the familiar Heisenberg bound  on  the standard deviation of the best unbiased measurement.   The argument is simple:  by  Chebyshev's inequality,  the standard deviation  $\sigma$ satisfies the bound $\sigma\ge \sqrt{(1-P)/4} \, \delta(P) $, which combined with the bound   $\delta (P) \ge c/n$   implies the Heisenberg scaling of the standard deviation.    It is important to stress   that  our ``Heisenberg-like'' bound  $\delta (P)  \ge c/n$ holds even when the measurement in question is not unbiased. 

  The size-accuracy bound (\ref{HL-memory}) can also be applied  to memories.   Suppose that one wants to write down the parameter $T$ with accuracy $\delta (P)$   into a quantum memory  of  $q$ qubits.      
  Then, Eq.  (\ref{HL-memory}) implies that, no matter what encoding is used,     the number of memory qubits must be a least 
  \begin{align}\label{HLformemory}
q    \ge  \log \frac 1 {\delta (P)}    +  O(1) \, .
\end{align}   We call Eq. (\ref{HLformemory}) the \emph{quantum memory bound}.  In the following, we show that  the bound  is tight, meaning that there exist  quantum states and quantum measurements for which Eq. (\ref{HLformemory}) holds with the equality sign.  Moreover, we show that these states can be generated from an ensemble of identically prepared  quantum particles by applying a  compression protocol that  minimises the memory size while preserving the accuracy.

\section{Compressing time information: the noiseless scenario}   
Consider a quantum clock made of $n$ identical particles
oscillating between two energy levels. Restricting the attention to these levels, each particle can be modelled as a qubit.     In the absence of noise, the evolution of each qubit is governed by  the Hamiltonian $H=E_0 \,|0\>\<0|+E_1 \, |1\>\<1|$, where $E_0$ and $E_1$ are the energy levels and  $|0\>$ and  $|1\>$ are the corresponding eigenstates.   For each  individual qubit, the best clock state is the  uniform superposition  $|+\>   =  (|0\>  +  |1\>  )/\sqrt 2$.    Choosing units such that   $(E_1- E_0)/\hbar  = 1$, the  state at time $T$ is  $|\psi_T\>    =   (|0\>  +   e^{-iT}  \,  |1\>)/\sqrt 2$. 

It is well known that $n$ qubits in an entangled state can achieve the Heisenberg scaling $1/n$ in terms of standard deviation,  from which it is immediate that the same scaling can be achieved in terms of inaccuracy \footnote{This can be shown by applying Markov inequality to the squared deviation from the mean.}.   
 However, here we consider $n$ qubits in a product state---specifically, the product state $|\psi_T\>^{\otimes n}$ obtained by preparing each qubit in the optimal single-copy state.    The state $|\psi_T\>^{\otimes n}$   has the standard scaling $\delta (P)  \approx 1/\sqrt n$ with the clock size (see Methods). And yet, it  can be compressed to a state that has the optimal scaling  with the memory size.  

To understand how  the compression works, it is useful to expand the clock  state  as 
\begin{align}\label{pureclock}
 |\psi_T\>^{\otimes n}  =  \sum_{k=0}^n  \,      e^{-i  k  T}  \,     \sqrt {B_{k,n,1/2}}  \,   |n,k\>  \, ,
\end{align}
  where $B_{k,n,1/2} $ is the binomial distribution with probability $1/2$, and $|n,k\>$ is the  state obtained by symmetrizing  the state $|1\>^{\otimes k} \otimes |0\>^{\otimes n-k}$ over the $n$ qubits.   The key observation is that,  for large $n$,  the binomial distribution is concentrated in an interval of size $ O(\sqrt n)$ around the average value  $\<k\>  = \lfloor n/2  \rfloor$. 
  This means that  the state  $ |\psi_T\>^{\otimes n}$  can be compressed into a typical  subspace of dimension $O( \sqrt n)$  without introducing significant errors.  More precisely, the errors are determined by the tails of the binomial distribution, which fall off exponentially fast as  $n$ increases.  
   After the clock state has been projected in the  typical subspace, it  can be encoded into $ 1/2 \log n$ memory qubits at the leading order.   This encoding attains the bound (\ref{HLformemory}): indeed,  the inaccuracy scales as $\delta (P)  \approx  1/\sqrt n $  for every fixed $P$, and the number of memory qubits, equal to  $ 1/2 \log n$,   grows exactly  as $\log  [1/\delta (P)]$.

The original state $|\psi_T\>^{\otimes n}$ can be retrieved from the compressed state, up to an error that vanishes exponentially fast with $n$.   Thanks to the exponential decay of the error,  a good compression performance  can obtained already for  small clocks:   for example,   $n=16$ is already in the asymptotic regime for all practical purposes. A  compression from 16  clock qubits to 4 memory qubits can be done with a compression error of $5.5\times10^{-3}$, in terms of the trace distance, or $3.0\times 10^{-5}$, in terms of the infidelity.  Relatively high quality compression can be obtained also for smaller number of qubits: for example, four clock  qubits can be encoded into  two memory qubits with fidelity $87.9\%$.

 Summarizing, the state of $n$ identically prepared clock qubits can be compressed into $1/2\log n$ memory qubits without compromising the accuracy. 
  The key ingredient of the compression protocol  is the projection  of the state (\ref{pureclock}) into the typical subspace spanned by energy eigenstates with oscillation frequencies in an interval of size $\sqrt n$ around the mean value. 
   We call this technique {\em frequency projection}.      

  \section{Extension to mixed states and noisy evolution}  We have seen that the optimal tradeoff between inaccuracy and memory size can be achieved for pure states with unitary time evolution.   A similar result can be obtained also in the noisy case. 
  Let us consider first the case where noise affects the state preparation,   while the evolution itself is still unitary.  In this model, each clock  qubit starts off in the mixed state  $\rho_{0,p}= p  |\psi_0\>\<\psi_0| +  (1-p)  I/2$  and evolves to the state  $\rho_{T,p}= p  |\psi_T\>\<\psi_T| +  (1-p)  I/2$.  Physically, we can think of these mixed states as the result of dephasing noise on the pure states   $|\psi_T\>$.   
  
In general, the amount of dephasing may vary from one qubit to another.   However,  as long as the variations are   random and affect all qubits equally and independently, the  state of the clock can be described as $\rho_{T,p}^{\otimes n}$, for some effective $p$.      Even more generally, one could consider some types of correlated  noise, where the errors acting on different qubits are part of  an (ideally infinite) exchangeable sequence \cite{kallenberg2006probabilistic}.    Physically, this means that each qubit undergoes a random phase kick, possibly correlated with the phase kicks received by the others, but without any systematic bias that makes one qubit more prone to noise than the others.  The model of exchangeable dephasing noise includes the correlated errors due to an overall uncertainty on the initial time of the evolution.  In general,  de Finetti's theorem  \cite{kallenberg2006probabilistic} implies that   exchangeable dephasing errors lead to a  mixture of states of the form $\rho_{T+ T_0,p}^{\otimes n}$, where  $T_0 $ is a random shift of the time origin and  $p$ is a random single-qubit dephasing parameter.     Thanks to this fact, we can focus first on the compression of the clock states $\rho_{T,p}^{\otimes n}$, and then include the case of correlated noise by allowing $p$ to vary.

The clock state  $\rho_{T,p}^{\otimes n}$ can be decomposed as a mixture of states with  different values of the total spin.     The decomposition is implemented by the Schur transform \cite{bacon-chuang-2006-prl}, which transforms the original $n$ qubit system into a tripartite system,  consisting of a spin register, a rotation register, and a permutation register, as in the following equation
 \begin{align}\label{decomp}
{\sf Schur}  \, \Big (   \rho_{T,p}^{\otimes n}  \Big)  =\sum_{J=0}^{n/2}q_{J,p}  \,  \Big(|J\>\<J|\otimes\rho_{T,p,J}\otimes\omega_J\Big) \, ,
\end{align}  
where $J$ is the quantum number of the total spin, $q_{J,p}$ is a probability distribution,  $  \{|J\>\}_{J=0}^{n/2}$ is an orthonormal basis for the spin register, $\rho_{T,  p, J}$ is a state of the rotation register, and $\omega_J$  is a fixed state of the permutation register, independent of $T$ and $p$. 

Since the states of the spin register are orthogonal, the value of $J$ can be read out without disturbing the state.   The problem is then to store the state $\rho_{T,p,J}\otimes\omega_J$    in the minimum amount of memory. Note also that the permutation register   can be discarded, for it contains no information about $T$.   Hence, the problem is actually to store the state $\rho_{T,p,J}$.   This can be done    through the technique of frequency projection, which is realised here by projecting the  state into the subspace  spanned by eigenstates of the total Hamiltonian in an interval of size $\sqrt J$ around the mean.

It turns out that the error introduced by frequency projection is negligible for large  $J$.      Specifically, we showed  that the trace distance between the original state and the frequency-projected state  is upper bounded as  
\begin{align}\label{Jbound}
\epsilon_{{\rm proj},J}\le (3/2)J^{{-\frac18\ln\left(\frac{p}{1-p}\right)}}+O\left(J^{-\frac18\ln J}\right)
\end{align}  
(see Supplementary Note 1 for the details).   The error of frequency projection becomes significant when $J$ is small, but fortunately the probability that $J$ is small tends to zero exponentially fast as $n$ grows: indeed, the probability distribution $q_J$ is the product  of a polynomial in $J$ times  a Gaussian with variance of order $\sqrt n$  centred  around the value $J_0  =   p(n+1)/2$ \cite{yang-chiribella-2016-prl}.

Frequency projection squeezes the density matrix into a subspace of size $\sqrt J$, which  can then be encoded into   a quantum memory of approximately $1/2 \log J$ qubits.  Now, the question  is whether $1/2 \log J$ is the minimum number of qubits  compatible with the quantum memory  bound (\ref{HLformemory}).  
  Here we show that the answer is affirmative for all values of $J$ in the high probability region.  
   The argument  is based on  two observations:  First, the inaccuracy for the state $\rho_{T,p}^{\otimes n}$, minimised over all possible measurements, has the  scaling $\delta_{\min} (P)   =   f(P)/\sqrt n$,  where $f(P)$  is a suitable function (see  Methods).   Second,  the state   $\rho_{T,p}^{\otimes n}$  can be  converted into any of the typical states $\rho_{J,T}$  in an approximately reversible fashion, with an error that vanishes in the large $n$ limit  \cite{universal}.   Since the inaccuracy is a continuous function of the state (see Methods),    we obtain that the minimum inaccuracy for the typical state  $\rho_{T,J}$ is
   $\delta_{\min}^{(J)}  (P)    =    {f(P)}/{\sqrt n}$  
   at leading order in $n$.   Now, recall that the typical values  of $J$ are equal to $J_0  =   p  (n+1)/2$, up to a correction of size at most $\sqrt n$.    Hence, one has  
    \begin{align}
   \delta_{\min}^{(J)}  (P)    =    \frac{  \sqrt{p}  \,  f(P)}{\sqrt {2 J}}  
    \end{align}
    at leading order.       Hence, the quantum memory bound (\ref{HLformemory}) implies that  at least $1/2 \log J $ memory qubits are necessary at leading order. But this is exactly the number of memory qubits used by our compression protocol. This concludes the proof that the protocol is optimal in terms of memory size for every typical value of $J$.

Note that the compression protocol does not require any knowledge of the time parameter $T$, nor it requires knowledge of  the dephasing parameter  $p$.  Thanks to this feature,  the protocol applies even in the presence of  randomly fluctuating and/or correlated dephasing, as long as   dephasing errors on different qubits arise from an  exchangeable sequence of random variables.

The protocol can also be applied when the evolution  is noisy.   Dephasing during the time evolution  is described by the  master equation  \cite{mixed-clock1},
  \begin{align}\label{master}
 \frac{\d \rho_t}{\d t}=  \frac i {\hbar} [\sigma_z,\rho_t]+\frac{\gamma}2(\sigma_z \rho_t\sigma_z-\rho_t) \, ,
 \end{align}
where  $\sigma_z$ is the  Pauli matrix  $\sigma_z  =  |0\>\<0|  - |1\>\<1|$ and $\gamma\ge0$ is the decay rate, corresponding to the inverse of the relaxation time $T_2$  in NMR.   The state at time $T$  is  
\begin{align}\label{noisy}
\rho_{T,\gamma}   =    p_{T,\gamma}  \,  |\psi_T\>\<\psi_T | +  (1-  p_{T,\gamma}) \,   \frac I2  \, ,
\end{align}
where $p_{T,\gamma}=  (1-e^{-\gamma (T+  \tau_0)})/2$, and $\tau_0$ accounts for dephasing noise in the state preparation.      The only difference with the case  of unitary evolution is that now the amount of dephasing depends on $T$.  
  However, since our compression protocol does not require knowledge of the dephasing parameter $p$, all the results shown before are still valid.

\section{Applications}
\subsection{Measuring the duration of a sequence of events.}   An important  feature of the compression protocol is that it is approximately reversible, meaning that the original $n$-qubit state can be retrieved from the memory, up to a small error that vanishes in the large $n$ limit (see Supplementary Note 2).   Thanks to this feature, one can engineer a setup  that pauses the time evolution and resumes it on demand.    
  
The setup, illustrated in Figure \ref{fig:flowchart},  uses a quantum clock made of $n$ identical qubits. 
 At time $t_0$, each qubit is initialized in the state $\rho_{t_0,\gamma}$. The qubit evolves until time $t_0'  =   t_0 +  T_0$  under the noisy dynamics (\ref{master}).    The state of the $n$ clock qubits  is then stored in the quantum memory, where it remains  until time $t_1$.  For simplicity, we assume that  the memory is ideal, meaning that the state of the memory qubits does not change during the lag time between $t_0'$ and $t_1$. Physically, this means that the decoherence time of the memory is long compared to the lag time between one event and the next.  At time $t_1$, the state of the memory is transferred back to the clock, which resumes  its evolution until $t_1'  =  t_1  +  T_1$. The procedure is iterated for $k$ times,  so that in the end the state of the clock records the total duration  $  T  =   T_0  +  T_1  +  \dots  +  T_{k-1}$. 
 
\begin{figure}[t!]
\centering
\includegraphics[width=0.45\linewidth]{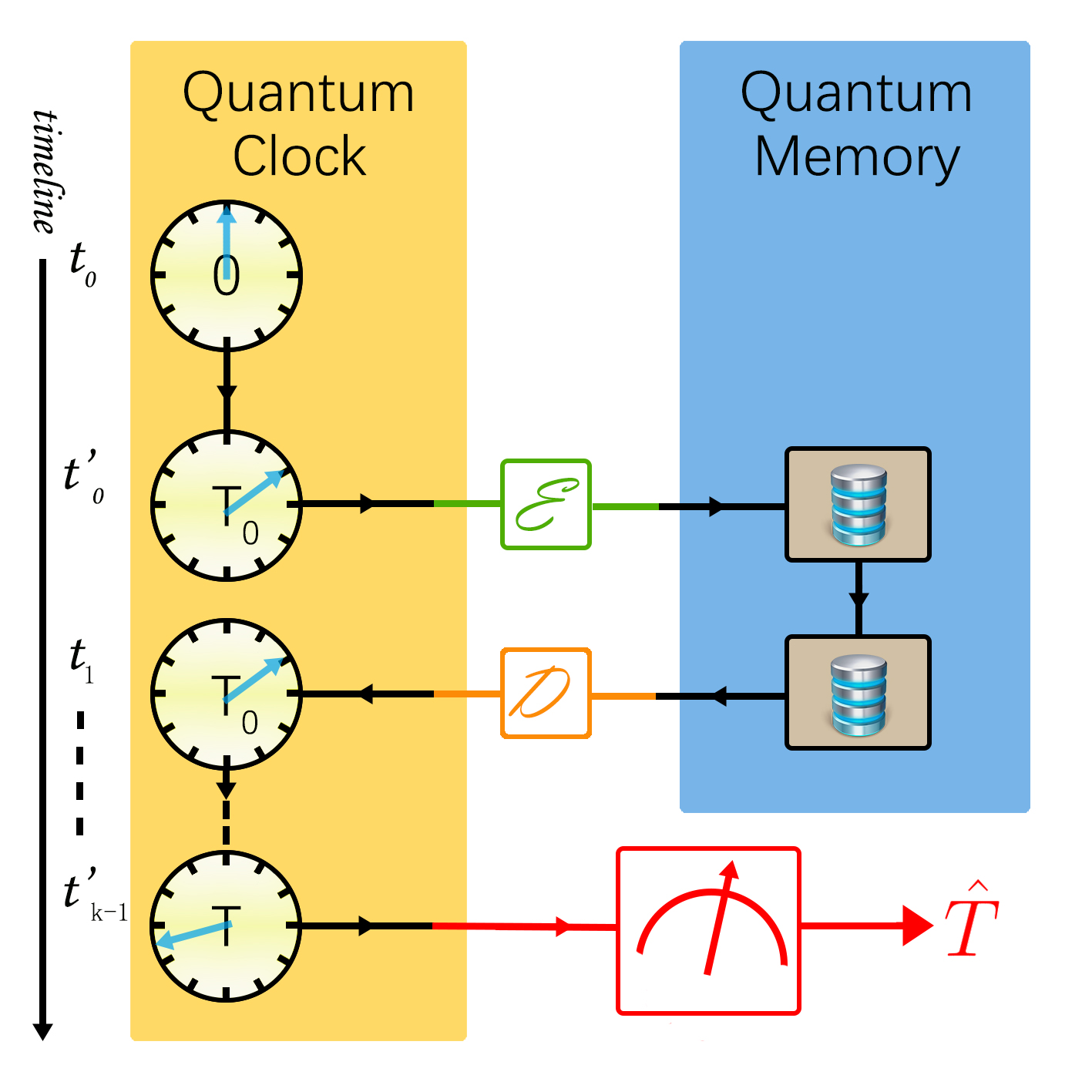}  
\caption{
{\bf  Coherent protocol measuring the total duration of $k$ events.}     
 The  clock starts its time evolution at time $t_0$ and continues evolving until time $t_0'$, when the first event is concluded. At this point, the time information is transferred to  the quantum memory, where it remains  until time $t_1$, when the  information is transferred back to the clock.   The procedure is repeated for $k$ times, so that the total duration of the $k$ events  is coherently recorded in the state of the memory.  Finally, the state of the memory is measured, yielding an estimate of the total duration.  } \label{fig:flowchart}
\end{figure}

Our coherent setup offers an advantage  over incoherent protocols where the duration of each time interval  $T_j$ is measured individually.   
In the noiseless scenario, the comparison is straightforward.  The probability distribution for the optimal  time measurement on the state $|\Psi_{T_j}\>^{\otimes n}$  is approximately Gaussian, and   the inaccuracy for the sum of $k$  Gaussian variables grows as $\sqrt k$. Instead, the inaccuracy of the coherent protocol is approximately constant in $k$, up to higher order terms arising from the compression error  (see Methods). Hence, performing a single measurement  reduces the inaccuracy  by a factor $\sqrt{k}$.   

The advantage of the coherent protocols persists even after taking into account the error of frequency projection, and even for relatively small $n$. As a simple example let us consider  the case of  $n=8$ and $P=0.9$. For a sequence of $k=4$ events, a coherent protocol using three qubits of memory has an inaccuracy $0.787$ times that of the incoherent protocol. 


The benefits of the coherent protocol  are not limited to the noiseless scenario. A performance comparison between coherent and incoherent protocols  is presented in Figure \ref{fig:comparison}, for the task of measuring the duration of a time interval $T$, divided into $k$ subintervals of equal length.   The figure shows the advantage of the coherent approach  for $\gamma  =  0.2$ and for every $k$ larger than $2$. 
\begin{figure}[h!]
\centering
\label{3D}\includegraphics[height=4.5cm]{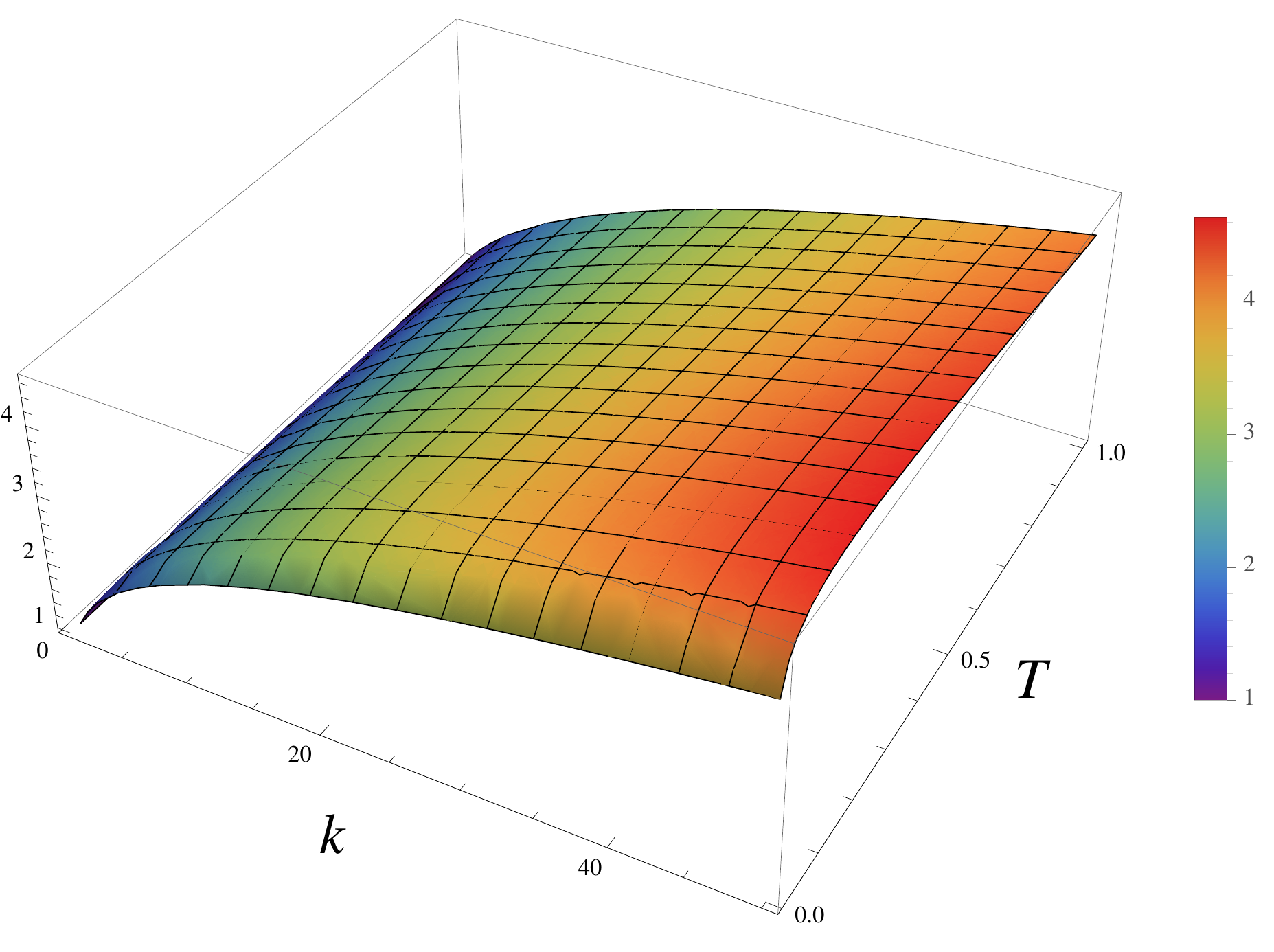} \hspace{0.5cm}  \label{MSE} \includegraphics[height=4.5cm]{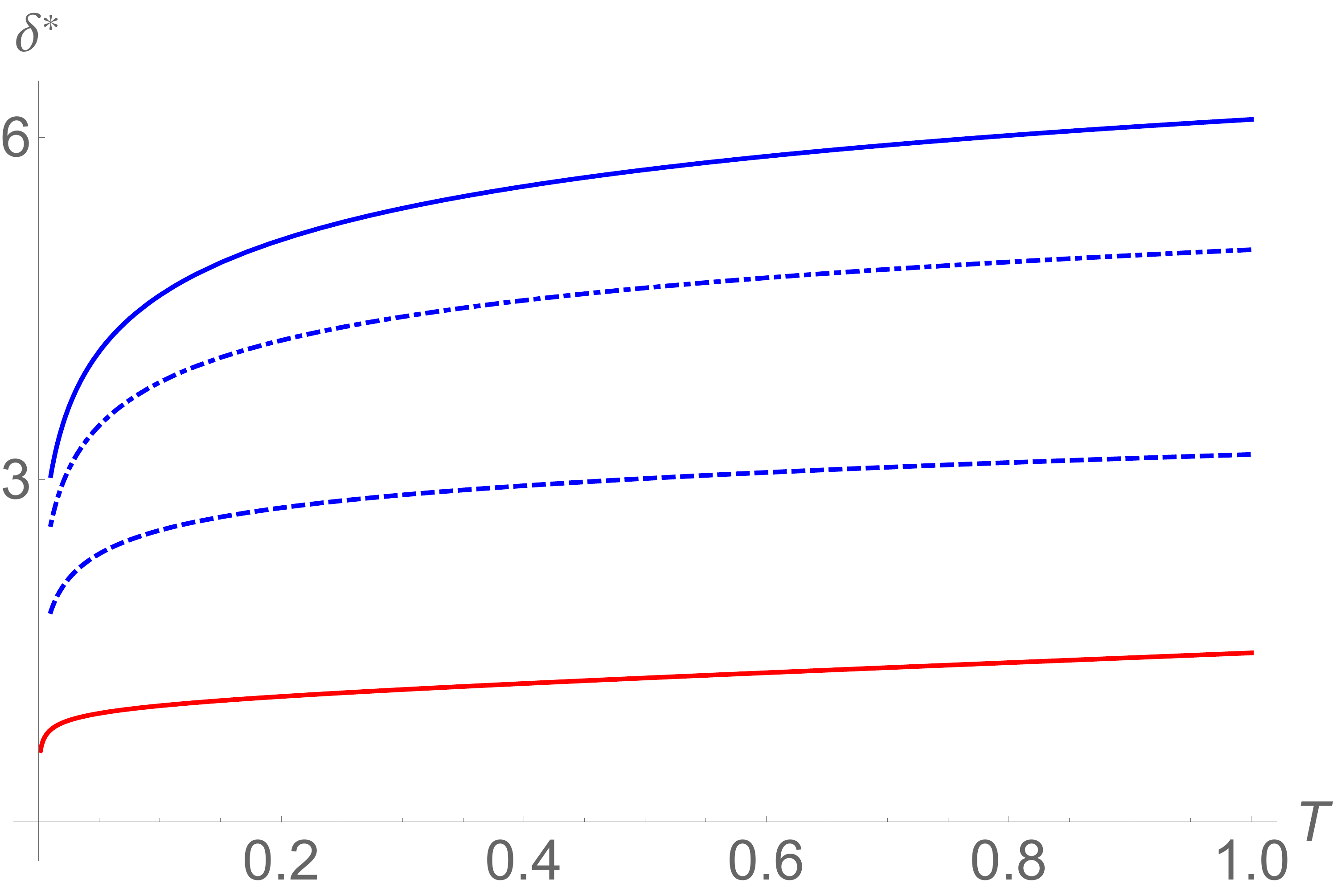}  
\caption{{\bf Comparison between coherent and incoherent protocols}.   {\em On the left:}    ratio between the inaccuracy of the  incoherent protocol  and the  inaccuracy of the coherent protocol.   The figure shows the ratio for clock qubits with decay rate $\gamma  =   0.2$.  The quantum advantage grows with the number of time intervals $k$, with an inaccuracy reduction of about 5 times for $k=50$.       {\em On the right:} Dependence of the rescaled  inaccuracy $\delta^*  = \sqrt n \,  \delta$ on   the total time $T$ for $\gamma  =  0.2$.  For the coherent protocol (red line) the inaccuracy is independent of $k$, while incoherent protocols  have inaccuracies increasing with $k$, illustrated by the blue lines    for $k  =  10, 30,$ and $50$. }\label{fig:comparison}
\end{figure}
  In Supplementary Note 3   we provide   a necessary and sufficient condition for the coherent protocol to have better performance than  the incoherent protocol.        The condition shows that the total duration is better computed coherently  whenever  the length of each subinterval $T/k$ is  small compared to the decoherence time $1/\gamma$.  Note that the advantage persists even when the total time $T$ is large, although the performance of both the coherent and the incoherent protocol worsen as  $T$ grows.

In the above discussion we assumed that the the memory is free from noise and that the compression protocol is implemented without errors.  Of course, realistic implementations will also involve errors.  One way to take  into account the noise in the memory  is to introduce an effective dephasing rate  $\gamma_j$, which models independent and symmetric errors occurring  in the lag time between $t_j'$ and $t_{j+1}$.    Overall, the result of this extra dephasing  is to reduce the size of the parameter region where the coherent storage of time information offers an advantage over the incoherent strategy.  
Now, let us consider the errors in the implementation of the compression protocol.    Thanks to the continuity of the inaccuracy (see Methods), the error of the circuit implementation can be analysed  independently of estimation inaccuracy.
Assuming an independent error model, the errors in the implementation of the encoding and decoding operations will introduce an error $\epsilon_{\rm circuit}$, which is bounded by the error probability of each elementary gate, denoted by  $\epsilon_1$, times the gate complexity of the whole circuit. The overhead of the gate complexity is the complexity of the Schur transform \cite{bacon-chuang-2006-prl}, which was recently reduced to $n^4\log n$ \cite{kirby2017practical}. For $k$ iterations, the overall error $\epsilon_{\rm circuit}$ scales as $k\epsilon_1 n^4\log n$, resulting in an additional term $\epsilon_{\rm circuit}/\sqrt{n}$ to the inaccuracy. Therefore, one can see that the inaccuracy will remain almost unaffected as long as the gate error of the compression circuit $\epsilon_1$ is small compared to $\left(kn^4\log n\right)^{-1}$. This is, of course, a challenging requirement, but it is important to note that the required gate error can be achieved  using fault tolerance,  without the need of implementing physical gates with an error vanishing with $n$.  In fact, the desired rate of $\epsilon_1$ can be achieved by using physical gates with error below a constant threshold value, by recursively increasing the number of layers of  error correction  \cite{aharonov1997fault,kitaev1997quantum,knill1998resilient}.

\subsection{Stabilizing quantum clocks in a network.}   Networks of clocks are important in many areas, such as GPS technology and distributed computing.  
Recently,  K\' om\'ar {\em et al}  proposed a quantum protocol,  allowing multiple nodes in a network to  jointly stabilize their clocks with higher accuracy \cite{network-nature,network-prl}.  
 The protocol involves a network  of $k$ nodes,  each node with a local oscillator used as a time-keeping device. The goal is to guarantee that all local oscillators have approximately the same frequency.     To this purpose, a central station distributes  a GHZ state $|{\tt GHZ}\>   =  \big (    |0\>^{\otimes k}  + |1\>^{\otimes k} \big)/\sqrt 2$ to the $k$ nodes.  The entanglement is then transferred to $k$ atomic clocks.    By interacting with the clock qubits, the $k$ nodes adjust the frequencies of their local oscillators, obtaining a shared time standard with accuracy  $1/(\sqrt n k)$, where $n$ is the number of repetitions of the whole procedure.   The key ingredient in this last step  is a protocol for estimating the sum of the frequencies of the local oscillators.  
 
Overall, the above protocol requires the communication of $  kn$  qubits.  Using our stopwatch protocol, we can reduce the amount of  quantum communication to $k  \log n  /2$  qubits at the leading order.    This can be done in two different ways.  The first way is to use the a sequential protocol, where the first node lets their local oscillator interact with the atomic clock for a fixed time $T_0$, and then encodes the state of the clock into a memory, which is sent to the second node.  The second node performs the same operations, and passes the memory to the third node, and so on until the $k$-th node.  In the end of the protocol, the memory will contain information about the total phase $\varphi_{\rm tot}  =  (\omega_1  +  \omega_2  +  \dots  +  \omega_k)  T_0$, where $\omega_j$ is the frequency of the $j$-th local oscillator. In this way, the  sum of the frequencies can be read out with an error  $1/\sqrt n$ independently of $k$, meaning that the average frequency has Heisenberg limited error of size $1/\left(\sqrt n  \, k\right)$.   

An alternative to the sequential protocol is to use a parallel protocol, where the central  station distributes entangled states to the $k$ nodes, as in Refs.    \cite{network-nature,network-prl}.
Using our frequency projection technique, it is possible to reduce the amount of quantum communication also in this case, from $kn$ qubits to $k  \log n  /2  $ qubits in total.  The idea is to compress the $n$ copies of the GHZ state into  a multipartite entangled state where each node has an exponentially smaller clock of  $1/2 \log n$ qubits.  Exploiting this fact,  one can obtain the same precision of  Refs. \cite{network-nature,network-prl} using an exponentially smaller amount of quantum communication between the nodes and the central station, as illustrated in Figure \ref{fig:network}.   
\begin{figure}[h!]
\centering
\label{3D}\includegraphics[height=5.5cm]{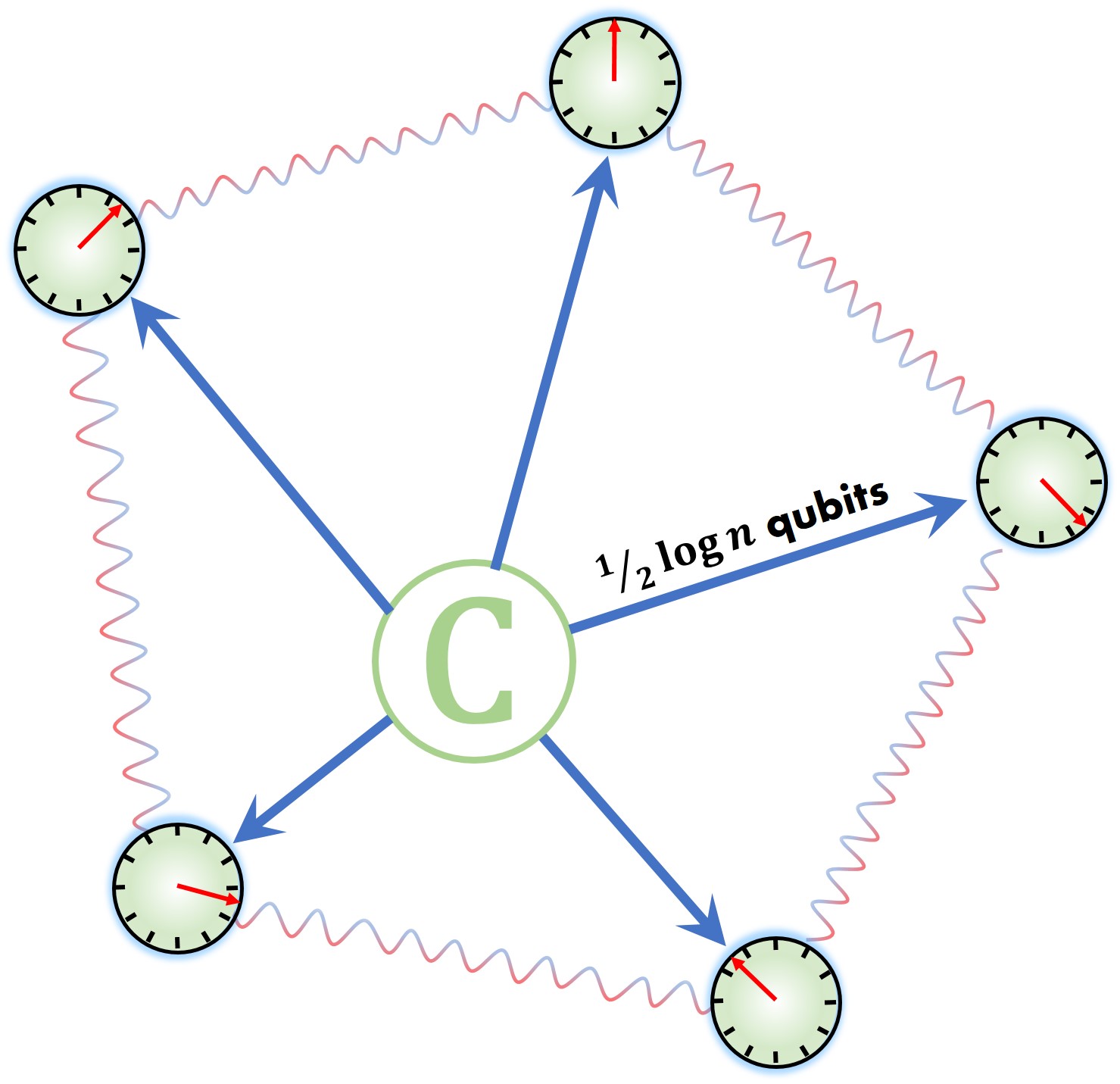}
\caption{{\bf Boosting the performance of quantum sensor networks.}   
A central station $\map C$ distributes quantum  information to  $k$ nodes, generating entanglement among $k$ local clocks.     Using frequency projection, the central station distributes a quantum state that guarantees the same precision of  $n$ GHZ states, while requiring an exponentially smaller amount of quantum communication of  $1/2\log n$ qubits per node. }\label{fig:network}
\end{figure}

\section{Conclusion}

The compression of clock states  is a versatile technique.  In the addition to advantages in measuring the total duration and in stabilizing quantum clocks in a network, it offers the opportunity  to transform product states of $n$ qubits into entangled states of $\sqrt n$ qubits, allowing one to reversibly switch from an encoding where the information can be accessed locally to a more compact encoding where the information is available globally.   
Quite interestingly, this approach works also for mixed states and in the presence of noise, thus defining a new set of mixed states achieving the Heisenberg limit.

In view of the applications,  it is natural to ask what ingredients would be needed to implement our compression protocol  experimentally. 
   The protocol requires a quantum computer capable of implementing the encoding and decoding operations.  The question is how large the computer should be and how many elementary operations it should perform.  In terms of size, we have seen that the large $n$ regime can be already probed for values  around $n=16$, a number that is  likely to be within reach in the near future. 
      In terms of complexity,  
       one can break down the compression protocol in two parts: the Schur transform and the frequency projection.    The Schur transform can be efficiently realised by a quantum circuit of at most polynomially many gates \cite{bacon-chuang-2006-prl}, scaling as $n^4\log n$ according to the most recent proposal \cite{kirby2017practical}. The circuit is simpler in the noiseless case \cite{buzek} and has been recently implemented in a prototype  photonic setup \cite{rozema}, which however is hard to scale to  larger number of qubits.  NMR and ion trap systems are another good candidate for prototype demonstrations of the Schur transform with small numbers of qubits, such as $n=10$.     The frequency projection   can be efficiently implemented with a technique introduced in Ref. \cite{buzek}, whereby the  spin eigenstate $|J,m\>$ is encoded into a $(2J+1)$-qubit state, with the $m$-th qubit is in the state $|1\>$ and all the other qubits are  in the state $|0\>$.   In this encoding, projecting on a restricted range of values of $m$ is the same as throwing away some of the qubits. 
The encoding operations and their inverses  can be implemented  using $O(J^2)$ elementary gates \cite{buzek}.   In summary, all the components of the quantum stopwatch can be implemented with a moderate amount of elementary gates. 
    The main challenge for the experimental realisation of our protocol is the required accuracy in the encoding and decoding operations, whose fault tolerant realisation requires additional layers of error corrections.  Our protocol provides an additional motivation to the realisation of fault-tolerant  quantum computers, showing  an example of application where the aid of a  quantum computer could significantly enhance the precision of time measurements. 

\section{Methods}  
\subsection{Properties of the inaccuracy}    The results of this paper take advantage of three basic properties of the inaccuracy, presented in the following:  \\
{\em (i) Continuity. }Suppose that the states $\rho_T$ and $\rho_T'$ are close in trace distance for every value of the parameter $T$. Operationally, this means that the  outcome probabilities for every measurement performed on $\rho_T$ are close to the outcome probabilities for the same measurement on $\rho_T'$.
  If the trace distance is smaller than $\epsilon$, then one has
  \begin{align}\label{bounds}
P(\delta, T )  - \epsilon  \le  P'(\delta,  T)  \le P(\delta  ,  T)  + \epsilon  ~ ,  
\end{align}
where $P(\delta,T)$   [respectively, $P'(\delta,T)$] is  the probability that the estimate falls within an interval of size $\delta$ around the true value  $T$.   In turn, Eq. (\ref{bounds}) implies the following bound in the inaccuracy  
\begin{align}
\delta  (P-\epsilon, T )    \le  \delta'(P,  T)  \le  \delta (  P+  \epsilon  ,  T)   ~ ,  
\end{align}
where $\delta(P,T)$ and $\delta' (P,T)$ are the inaccuracies for the states $\rho_T$ and $\rho_T'$, respectively.   When the probability distribution is sufficiently regular, these bounds guarantee that the inaccuracy is a continuous function of the state.  For example, we will see that the accuracy for an $n$-qubit quantum clock in the state $\rho_{T,\gamma}^{\otimes n}$ is equal to $\delta (P,  T)  = f(P)/\sqrt n $ at the leading order, where $f(P)$ is an analytical function of $P$.    A state that is $\epsilon$-close to  $\rho_{T,\gamma}^{\otimes n}$ will also have accuracy  $ \delta' ( P)  =  f(P)/\sqrt n$, up to a correction of size   $\epsilon/\sqrt n$.  In the case of our compression protocol, the compression protocol vanishes with $n$, meaning that the correction does not affect the leading order.  

{\em (ii)   Data-processing inequality.}  
Suppose that the system $S$, used to encode the parameter $T$, is transformed into another system $S'$ by some physical process.  Let $\rho_T'$ be the state generated by the process acting on the state $\rho_T$.    Then,  every measurement $M'$ on the output system  $S'$ defines a measurement  $M$ on the input system $S$, obtained by first transforming $S$ into $S'$ and then performing the measurement $M'$.  By construction, the two measurements have the same statistics and, in particular, the same inaccuracy.   
  By minimising the inaccuracy over all measurements, one obtains the inequality  
\begin{align}
 \delta'_{\min}   (P) 
\ge \delta_{\min}    (P)     \,, 
\end{align}
expressing the fact that physical processes cannot reduce the minimum inaccuracy.  

{\em (iii)  Symmetry.} For quantum clocks,  the accuracy is maximised by covariant measurements \cite{holevo}, that is, measurements  described by  positive operator valued measures of the form $M_{\widehat T}  = e^{-i \widehat T  H}  M  e^{i \widehat TH} $.       
  For covariant measurements,   one can show  that mixing increases the inaccuracy: 
\begin{prop} For a convex mixture $\rho  =  \sum_i  p_i\, \rho_i$, one has the inequality $\delta_{\min}  ( P)  \ge  \min_i \delta^{(i)}_{\min}  (P)$, where $\delta_{\min} (P)$   [respectively, $\delta_{\min}^{(i)}  (P)$] is  the minimum inaccuracy for the state $\rho$ [respectively, $\rho_i$]. 
\end{prop}  
The argument is simple.  For every covariant measurement, the probability to find the estimate in an interval of size $\delta$ around the true value $T$ is independent of  $T$ and will be denoted simply  as $P(\delta)$.     For measurements on the state $\rho  =  \sum_i  p_i\, \rho_i$, the probability    is the convex combination  $P( \delta)  =  \sum_i  p_i\,  P_i(\delta)$, where $P_i (\delta)$ is the corresponding probability for the state $\rho_i$.    Setting $\delta  =  \min_i \delta_{\min}^{(i)}  (P)$,  one has the inequality 
$P(\delta)  \le \sum_i   p_i   P    =  P$, with equality if and only if all the inaccuracies $\delta^{(i)}_{\min} (P)$ are equal.    
Hence, the inaccuracy for the mixture  $\rho$ cannot be smaller  than $\delta$.

\subsection{Accuracy of time measurements on identically prepared clock states}   We now construct a measurement strategy that estimates the parameter $T$ from  the state $\rho_{T,p}^{\otimes n}$ with inaccuracy $\delta (P) =  O( 1/\sqrt n)$.     In this strategy,   each clock qubit is measured independently, projecting  the $j$-th   qubit  on the eigenstates   of the observable $O_j   =    \cos \tau_j  \,  \sigma_x  +   \sin  \tau_j  \,  \sigma_y$, where  $\tau_j$ is an angle, chosen uniformly at random between 0 and $2\pi$.       The measurement has two possible outcomes, $+1$ and $-1$.  If the outcome of the measurement  is $+1$,  one records the value  $\tau_j$,  if the outcome is $-1$, one records the value   $\tau_j+\pi$.     Mathematically,  the measurement strategy is described by the positive operator valued measure     $\{M_\tau\}_{ \tau \in [0,2\pi) }$  with 
\begin{align}\label{cov}
M_{\tau} =
\frac{1}{2\pi}
\left(
\begin{array}{cc}
1 & e^{-i\tau}\\
e^{i\tau} & 1
\end{array}
\right).
\end{align}
The probability density  that the measurement  yields the outcome $\tau$ when the input state is $\rho_{T,p}$ is
$P(\tau|  T, p) =
\Tr \left[ M_{\tau} \rho_{T,p}\right]$. Explicit calculation shows that the classical Fisher information of this probability distribution is  
  \begin{align}\label{floc} 
  F_{\rm loc}   &=     1-2\sqrt{p(1-p)}.
 \end{align}       
Now, since $n$ qubits are measured independently, one can collect the results of the measurements and use classical statistics to generate an estimate of  $T$.     
  Using the maximum likelihood estimator, one obtains an estimate that is approximately Gaussian-distributed  with average $T$ and standard deviation  $\sigma  =  1/\sqrt {  n  F_{\rm loc}}$.    The  probability that the estimate $\widehat{T}$ deviates from the true value by less than $\delta/2$ is \cite{van-book}
\begin{align}\label{idealerror}
\mathsf{Prob} \left(|\widehat{T}-T|\le \frac{\delta}2 \right)=&
\erf\left(   \delta  \,    \sqrt{\frac{nF_{\rm loc}  }{8  }}  \right) +   O\left(\frac 1{\sqrt n} \right) 
 \, ,  
\end{align}
where  $\erf(x)$
  is the error function. Hence,  the inaccuracy can be expressed as 
\begin{align}
\delta_{\rm loc} (P)=\sqrt{\frac 8{nF_{\rm loc}}}  \, \erf^{-1}(P)  +   O\left(\frac 1{n} \right)   \, .
\end{align}

The same argument applies to the states $\rho_{T,\gamma}$, generated by the noisy time evolution.   The only difference is that, in this case, the classical Fisher information is 
\begin{align}\label{floc} 
  F_{\rm loc}   &=     1-\gamma^2-\sqrt{1-e^{-2\gamma T}} +\frac{\gamma^2}{\sqrt{1-e^{-2\gamma T}}}  ,
 \end{align}          
when the decay rate $\gamma$ is known, and 
\begin{align}\label{localMSE}
F_{\rm loc}&={1-\sqrt{1-e^{-2\gamma T}}} 
\end{align}
when   $\gamma$  is unknown (see Supplementary Note 4 for the derivation).  

\vskip6pt

\noindent{\bf Acknowledgement.} We thank  Lorenzo Maccone for useful comments and Xinhui Yang for drawing the figures. This work is supported by the Hong Kong Research Grant Council through Grant No. 17326616 and 17300317, by National Science
Foundation of China through Grant No. 11675136, by the HKU Seed Funding for Basic Research,  the John Templeton Foundation, and by the Canadian Institute for Advanced Research (CIFAR).   YY is
supported by a Microsoft Research Asia Fellowship and a Hong Kong and China Gas Scholarship.
MH was supported in part by a MEXT Grant-in-Aid for Scientific Research (B) No. 16KT0017, a MEXT Grant-in-Aid for Scientific Research (A) No. 23246071,
the Okawa Research Grant, and Kayamori Foundation of Informational Science Advancement. Centre for Quantum Technologies is a Research Centre of Excellence funded by the Ministry of Education and the National Research Foundation of Singapore. 

\medskip



\bibliographystyle{RS}

\bibliography{ref} 
\appendix

\subsection{Supplementary Note 1: error bound for frequency projection [Eq. (6) of the main text].}

Here we prove Eq. (6) of the main text, regarding the error of frequency projection of the compressor. We prove the result in a general setting, where the states to be compressed are of the form 
\begin{align}\label{clock-define}
\rho_{\phi,p}&:=p\, |\phi\>\<\phi|+(1-p)\, |\phi_{\perp}\>\<\phi_{\perp}| \, ,
\end{align}
with $p\in(1/2,1]$, $|\phi\>=\sqrt{s}|0\>+\sqrt{1-s}e^{-i\phi}|1\>$, and $|\phi_{\perp}\>:=\sqrt{1-s}|0\>-\sqrt{s}e^{-i\phi}|1\>$ for some fixed $s\in(0,1)$. 
Note that the clock states    $\rho_{t,\gamma}$ considered in the main text are  a special case of states of the form (\ref{clock-define}), with $p=(1+e^{-\gamma t})/2$ and $\phi=t$.

To begin with, we recall a few basic facts from the main text. First, the $n$-fold clock state $\rho_{\phi,p}^{\otimes n}$ can be decomposed as 
\begin{align}\label{decomp}
\rho_{\phi,p}^{\otimes n}\simeq\sum_{J=0}^{n/2}q_J\left(|J\>\<J|\otimes\rho_{\phi,p,J}\otimes\frac{I_{m_J}}{m_J}\right) \, , 
\end{align}
where $\simeq$ denotes the unitary equivalence implemented by the Schur transform, $J$ is the quantum number of the total spin, $q_{J,t,\gamma}$ is a probability distribution,  $|J\>$ is the state of the index register,  $\rho_{t,\gamma,J}$ is the state of the representation register,  and  $I_{m_J}/m_J$  is the maximally mixed state in a suitable subspace of the multiplicity register    \cite{fulton-harris,book-hayashi}.  
The state $\rho_{\phi,p,J}$ can be expressed in the form 
\begin{align}
\rho_{\phi,p,J}=U_\phi^{\otimes n}\left(\rho_{p,J}\right)U_\phi^{\dag \, \otimes n}\qquad U_\phi=|0\>\<0|+e^{i\phi}|1\>\<1|,
\end{align}
where the fixed state  $\rho_{p,J}$ has the form
\begin{align}
&\rho_{p,J}:=(N_{J})^{-1}\sum_{m=-J}^{J}p^{J+m}(1-p)^{J-m}|J,m\>_s\<J,m|_s\\
&N_J:=\sum_{k=-J}^{J}p^{J+m}(1-p)^{J-m}
\end{align}
where $|J,m\>_s$ is the orthonormal basis defined as
\begin{align}\label{sym}
|J,m\>_s:=\frac{\sum_{\pi\in\grp{S}_{2J}}V_\pi |\phi_0\>^{\otimes (J+m)}|\phi_{0,\perp}\>^{\otimes (J-m)}}{\sqrt{(2J)!(J+m)!(J-m)!}}
\end{align}
with $|\phi_0\>=\sqrt{s}|0\>+\sqrt{1-s}|1\>$, $|\phi_{0,\perp}\>=\sqrt{1-s}|0\>-\sqrt{s}|1\>$, $\grp{S}_{2J}$ being the $(2J)$-symmetric group and $V_\pi$ being the unitary implementing the permutation $\pi$. 

The frequency projection channel $\map{P}_{{\rm proj}, J}$ is defined as
\begin{align}
&\map{P}_{{\rm proj},J}(\rho):=P_{{\rm proj},J}\,\rho\,P_{{\rm proj},J}+\left(1-\Tr[\rho\,P_{{\rm proj},J}]\right)\rho_0\nonumber \\
&P_{{\rm proj},J}:=\sum_{|m-(2s-1)J|\le \frac{\sqrt{J}\log J}2}|J,m\>\<J,m|
\label{define-setC},
\end{align}
where $\rho_0$ is a fixed state of the representation register.
What we need to prove is exactly the following theorem.
\begin{theo}\label{lemma-proj}
For large $J$, the frequency projection error $\epsilon_{{\rm proj},J}:=\frac12\left\|\map{P}_{{\rm proj},J}(\rho_{\phi,p,J})-\rho_{\phi,p,J}\right\|_1$ is upper bounded as
\begin{align}
\epsilon_{{\rm proj},J}\le (3/2)J^{{-\frac18\ln\left(\frac{p}{1-p}\right)}}+O\left(J^{-\frac18\ln J}\right)
\end{align}
 for every $t$.
\end{theo}
\medskip
The property below is useful in our proof of the error bound.
\begin{lem}\label{lemma-prop}
The basis $\{|J,m\>_s\}$ defined by Eq. (\ref{sym}) satisfies the property
\begin{align}\label{inner}
|\<J,m|_s|J,k\>|\le\left\{\begin{matrix}\sqrt{{2J\choose J+k}{2J\choose J-m}s^{2J+k-m}(1-s)^{m-k}},&\qquad s\ge 1/2\\
\\
\sqrt{{2J\choose J+k}{2J\choose J-m}s^{m+k}(1-s)^{2J-m-k}},&\qquad s< 1/2.\end{matrix}\right.
\end{align}
where $$|J,k\>:=\frac{\sum_{\pi\in\grp{S}_{2J}}V_\pi |0\>^{\otimes (J+k)}|1\>^{\otimes (J-k)}}{\sqrt{(2J)!(J+k)!(J-k)!}}$$ is the symmetric basis.
\end{lem}
\noindent{\it Proof.} Exploiting the symmetry of both bases, we have the following chain of inequalities.
\begin{align}
|\<J,m|_s|J,k\>|&=\frac{1}{(2J)!}\left|\sum_{\pi,\pi'\in\grp{S}_{2J}}\frac{\<\phi_0|^{\otimes(J+m)}\<\phi_{0,\perp}|^{\otimes(J-m)}V_{\pi}V_{\pi'}|0\>^{\otimes(J+k)}|1\>^{\otimes(J-k)}}{\sqrt{(J+m)!(J-m)!(J+k)!(J-k)!}}\right|\\
&=\left|\sum_{\pi\in\grp{S}_{2J}}\frac{\<\phi_0|^{\otimes(J+m)}\<\phi_{0,\perp}|^{\otimes(J-m)}V_{\pi}|0\>^{\otimes(J+k)}|1\>^{\otimes(J-k)}}{\sqrt{(J+m)!(J-m)!(J+k)!(J-k)!}}\right|\\
&=\sqrt{\frac{(J+m)!(J-m)!}{(J+k)!(J-k)!}}\left|\sum_{l=\max\{0,-k-m\}}^{\min\{J-m,J-k\}}(-1)^l{J-k\choose l}{J+k\choose J-m-l}s^{\frac{2l+k+m}2}(1-s)^{\frac{2J-m-k-2l}{2}}\right|\\
&\le\sqrt{\frac{(J+m)!(J-m)!}{(J+k)!(J-k)!}}\sum_{l}{J-k\choose l}{J+k\choose J-m-l}s^{\frac{2l+k+m}2}(1-s)^{\frac{2J-m-k-2l}{2}}\\
&=\sqrt{\frac{(J+m)!(J-m)!s^{J+k}(1-s)^{J-k}}{(J+k)!(J-k)!}}\sum_{l}{J-k\choose l}{J+k\choose J-m-l}\left(\frac{s}{1-s}\right)^{l-\frac{J-m}2}\\
&\le\sqrt{{2J\choose J+k}s^{J+k}(1-s)^{J-k}}\cdot\sqrt{{2J\choose J-m}\left(\frac{s}{1-s}\right)^{J-m}},
\end{align}
and thus we have reached Eq. (\ref{inner}). Notice that the last inequality is for the case $s\ge 1/2$. For $s<1/2$, the last term in the second square-root should be replaced by $[(1-s)/s]^{J-m}$.

\qed
\medskip

\noindent{\it Proof of Theorem \ref{lemma-proj}.}
First observe that since $\left[U_\phi^{\otimes n},P_{{\rm proj},J}\right]=0$, the error is independent of $\phi$. Then we can rewrite the error as 
\begin{align*}
\epsilon_{{\rm proj},J}\le&\frac12\left\{\left\|P_{{\rm proj},J}\,\rho_{p,J}\,P_{{\rm proj},J}-\rho_{p,J}\right\|_1+\left(1-\Tr\left[P_{{\rm proj},J}\,\rho_{p,J}\right]\right)\right\}\\
\le&\frac12\left\{2\sqrt{\left(1-\Tr\left[P_{{\rm proj},J}~\rho_{p,J}\right]\right)}+\left(1-\Tr\left[P_{{\rm proj},J}~\rho_{p,J}\right]\right)\right\}\\
\le&\frac{3}{2}\sqrt{1-\Tr\left[P_{{\rm proj},J}~\rho_{p,J}\right]},
\end{align*}
having used the gentle measurement lemma \cite{gentle,wilde} to derive the second last inequality.
We now simply have to estimate $1-\Tr\left[P_{{\rm proj},J}~\rho_{p,J}\right]$. First, we expend the quantity as
\begin{align*}
1-\Tr\left[P_{{\rm proj},J}\,\rho_{p,J}\right]=&\sum_{m=-J}^{J}\frac{p^{J+m}(1-p)^{J-m}}{N_{J}}\<J,m|_s(I- P_{{\rm proj},J})|J,m\>_s\\
=&\sum_{m=-J}^{J}\frac{p^{J+m}(1-p)^{J-m}}{N_{J}}\sum_{|k-(2s-1)J|> \frac{\sqrt{J}\log J}2}|\<J,m|_s|J,k\>|^2.
\end{align*}

Applying Lemma \ref{lemma-prop} (for $s\ge 1/2$, and the case $s<1/2$ can be treated in the same way), we have
\begin{align*}
1-\Tr\left[P_{{\rm proj},J}\,\rho_{p,J}\right]\le&\sum_{m=J-a}^{J}\frac{p^{J+m}(1-p)^{J-m}}{N_{J}}\left(\frac{s}{1-s}\right)^{J-m}{2J\choose J-m}\sum_{|k-(2s-1)J|> \frac{\sqrt{J}\log J}2}s^{k}(1-s)^{2J-k}{2J\choose k}\\
&\qquad+\sum_{m=-J}^{J-a-1}\frac{p^{J+m}(1-p)^{J-m}}{N_{J}}
\end{align*}
where $a\ge 2$ is a parameter to be specified later. We continue bounding the term as
\begin{align}
1-\Tr\left[P_{{\rm proj},J}~\rho_{p,J}\right]\le&\left(\frac{s}{1-s}\right)^a {2J\choose a}\sum_{|k-(2s-1)J|> \frac{\sqrt{J}\log J}2}s^{k}(1-s)^{2J-k}{2J\choose k}+\left(\frac{1-p}{p}\right)^{a+1}\\
\le& \frac{2(2J)^a}{a!}\exp\left[-\frac{\log^2 J}{4}+a\ln\left(\frac{s}{1-s}\right)\right]+\left(\frac{1-p}{p}\right)^{a+1}\\
\le &\exp\left[-\frac{\ln^2 J}{4\ln^2 2}+a\ln (2J)+a\ln\left(\frac{s}{1-s}\right)\right]+\left(\frac{1-p}{p}\right)^{a+1},\label{bound-term}
\end{align}
having used Hoeffding's bound and $a\ge 2$. 
Choosing, for instance, $a=\lfloor(\ln J)/4\rfloor$ guarantees that 
\begin{align*}
\epsilon_{{\rm proj},J}&\le\frac{3}{2}\sqrt{\exp\left[-\frac{\ln^2 J}{4\ln^2 2}+\frac{\ln^2 J}{4}+\frac{\ln J}{4}\ln\left(\frac{s}{1-s}\right)\right]+\left(\frac{1-p}{p}\right)^{(\ln J)/4+1}}\\
&\le\frac{3}{2}\left(\frac{1-p}{p}\right)^{(\ln J)/8}+O\left(J^{-\frac18\ln J}\right)\\
&= (3/2)J^{-\frac18\ln\left(\frac{p}{1-p}\right)}+O\left(J^{-\frac18\ln J}\right).
\end{align*}

\qed

\subsection{Supplementary Note 2: bound on the storage error.}\label{app:inner}

\medskip

Here we prove the compressor mentioned in the main text has a vanishing error. First, notice that the  trace distance between the original state $\rho_{\phi,p}^{\otimes n}$ and the output of the compression protocol, denoted as $\rho_{\phi,p,n}'$, is 
\begin{align}
\epsilon_{\phi,p}&=\frac12\left\|\rho'_{\phi,p,n}-\rho_{\phi,p}^{\otimes n}\right\|_1\nonumber\\
&=\frac12\left\|\sum_J q_J|J\>\<J|\otimes\left[\map{P}_{{\rm proj},J}(\rho_{\phi,p,J})-\rho_{\phi,p,J}\right]\otimes\frac{I_{m_J}}{m_J} \right\|_1\nonumber\\
&=\sum_J q_J\cdot\epsilon_{{\rm proj},J},\label{easy-error1}\qquad\qquad\epsilon_{{\rm proj},J}=\frac12\left\|\map{P}_{{\rm proj},J}(\rho_{\phi,p,J})-\rho_{\phi,p,J}\right\|_1.
\end{align}
Since we already have a bound for $\epsilon_{{\rm proj},J}$, the only ingredient needed for the error bound is the concentration property of $q_J$. Notice that the probability distribution $q_J$ in Eq. (\ref{decomp}) has the explicit form \cite{universal}
\begin{align}
q_{J}=\frac{2J+1}{2J_0}&\left[  B\left(\frac n2 +  J+1 \right) -B\left(\frac n2 - J \right)\right]\label{qJ}
\end{align}
where $B(k)=p^{k}(1-p)^{n-k}{n\choose k}$ and $J_0 = (p-1/2)(n+1)$, which is a Gaussian distribution concentrated in an interval of width $O(\sqrt{n})$ around $J_0$ when $n$ is large. Picking an interval much larger than $\sqrt{n}$, we get from Hoeffding's inequality that 
\begin{align}
\sum_{|J-J_0|\le n^{2/3}}q_J\ge 1-2\exp\left[\frac{2n^{1/3}}{p^2}\right].
\end{align}
Substituting the above inequality into Eq. (\ref{easy-error1}), we have
\begin{align*}
\epsilon_{\phi,p}&\le \max_{|J-J_0|\le n^{2/3}}\epsilon_{{\rm proj},J}+\sum_{|J-J_0|> n^{2/3}} q_J\\
&\le\frac32 \left(\frac{2}{(2p-1)n}\right)^{\frac18\ln\left(\frac{p}{1-p}\right)}+O\left(n^{-\frac18\ln n}\right).
\end{align*}
Finally, substituting $p=(1+e^{-\gamma t})/2$ into the above inequality, the error for the compressor in the depolarizing model is bounded by
\begin{align*}
\epsilon_{\gamma,t}&\le\frac32 \left(\frac{2e^{\gamma t}}{n}\right)^{\frac18\ln\cosh\frac{\gamma t}{2}}+O\left(n^{-\frac18\ln n}\right).
\end{align*}

From the above bound, we can directly get the bound for the overall error of compression in a quantum stopwatch as
\begin{align}\label{error}
\epsilon^{(k)}_{T,\gamma}\le  \frac{3k}2 \left(\frac{2e^{\gamma T}}{n}\right)^{\frac18\ln\coth\frac{\gamma T}{2}}  +  O\left(k  n^{-\frac18\ln n}\right)  \, ,
\end{align}
where $k$ is the number of events and $T$ is the total duration.

\subsection{Supplementary Note 3: coherent versus incoherent protocols.}
Here we compare the accuracy of the coherent protocols with the accuracy of incoherent protocols using repeated measurements. 

Suppose that we  want to measure the total duration of  $k$  intervals, each of which has length $T/k$. 
In the coherent protocol, each clock qubit starts with a pure state and ends up in a mixed state with maximum eigenvalue depending on the total duration $T$ as  $p(T,\gamma)=(1+e^{-\gamma T})/2$. 
The inaccuracy for the quantum stopwatch thus has leading order $1/\sqrt{nF_{\rm loc}}$, where   \begin{align}
\nonumber  F_{\rm loc}   =   & 1-\gamma^2-\sqrt{1-e^{-2\gamma T}}+\frac{\gamma^2}{\sqrt{1-e^{-2\gamma T}}}\, 
 \end{align}   
 is the classical Fisher information.   On the other hand, an incoherent protocol involves $k$ measurements of time and  $k$ initializations of the clock qubits,  resulting in an  of leading order $\sqrt{k/(nF'_{\rm loc})}$, where $F'_{\rm loc}$ is the classical Fisher information of the individual  step, given by
  \begin{align}
\nonumber  F'_{\rm loc}   =   &   1-\gamma^2-\sqrt{1-e^{-2\gamma T/k}}+\frac{\gamma^2}{\sqrt{1-e^{-2\gamma T/k}}} \, .
 \end{align} 
 
Comparing the two error terms, we conclude that the coherent protocol outperforms the incoherent  one when $F_{\rm loc}\ge (F'_{\rm loc})/k$, namely that
\begin{align*}
k\left[ 1-\gamma^2-\sqrt{1-e^{-2\gamma T}}+\frac{\gamma^2}{\sqrt{1-e^{-2\gamma T}}}\right]\ge  1-\gamma^2-\sqrt{1-e^{-2\gamma T/k}}+\frac{\gamma^2}{\sqrt{1-e^{-2\gamma T/k}}}.
\end{align*}
Finally we comment on the case of large $k$, fixing $\gamma$ and $T$. In this case, the left hand side of the above inequality scales linearly in $k$, while the right hand side scales as $\sqrt{k}$ (which can be seen by taking the Taylor expansion of the last term). Therefore, the left hand side dominates the right hand side in the large $k$ limit, and thus the coherent protocol is always better for large $k$.

\subsection{Supplementary Note 4:  inaccuracy of the local strategy for unknown $\gamma$.}
Here we analyze the precision of local time measurement strategy when $\gamma$ is unknown and is treated as a nuisance parameter, which is not directly of interest but affects the analysis of our estimation.
For this analysis, we parameterize the distribution of our interest as $P_{T,\gamma}(\tau)$
because the dependence on $\gamma$ is crucial in the following discussion.
Then, we introduce Fisher information matrix $ F_{\rm loc}$ defined by
$$\left( F_{\rm loc}\right)_{x,y}:=\int_0^{2\pi}d\tau~ 
\frac{\partial \log P_{T,\gamma}(\tau)}{\partial x}\cdot
\frac{\partial \log P_{T,\gamma}(\tau)}{\partial y}\,P_{T,\gamma}(\tau),$$
where $x,y\in\{T,\gamma\}$. And we need to evaluate $\left( F_{\rm loc}^{-1}\right)_{TT}$, namely the $(T,T)$-entry of the inverse Fisher information matrix $ F_{\rm loc}^{-1}$.


To simplify the calculation, we first adapt a re-parameterization of the nuisance parameter from $\gamma$ to $p=(e^{-\gamma T}+1)/2$. 
The $T$ component $\left( F_{\rm loc}^{-1}\right)_{TT}$ of the inverse Fisher information is invariant under such a re-parameterization of the nuisance parameter, which can be proved as follows.

The inverse Fisher information $ F_{\rm loc}^{-1}$ with respect to the parametrization $(T,\gamma)$  is given by
$$  F_{\rm loc}^{-1}=\left[\begin{matrix} \left(  F_{\rm loc}\right)_{TT} & \left(   F_{\rm loc}\right)_{T,\gamma} \\ \left(  F_{\rm loc}\right)_{\gamma,T} & \left(   F_{\rm loc}\right)_{\gamma \gamma}\end{matrix}\right]^{-1}.$$
Now, consider a re-parameterization of the nuisance parameter $\gamma\to p$, and the new Fisher information matrix is
$$ F_{{\rm loc},p}^{-1}=A^T   F_{\rm loc}^{-1} A$$
where $$A=\begin{pmatrix} 1 & 0 \\ \frac{\partial \gamma}{\partial T} & \frac{\partial \gamma}{\partial p}\end{pmatrix}$$ is the Jacobi matrix. Therefore, we have 
\begin{align*}
 F_{\rm loc}^{-1}=A\left( F_{{\rm loc},p}^{-1}\right)A^T
\end{align*}
and it can be checked by straightforward calculation that 
$$ \left(  F_{\rm loc}^{-1}\right)_{TT}= \left( F_{{\rm loc},p}^{-1}\right)_{TT}.$$

It is therefore enough to calculate the Fisher information matrix for the parameterization $(T,p)$, when the input state is  $\rho_{T,p}=p |\varphi_T\>\<\varphi_T|+\left(1-p\right)\, |\varphi_{T}^\perp\>\<\varphi_{T}^\perp|$ with $|\varphi_{T}\>=(|0\>+e^{-iT}|1\>)/\sqrt{2}$ and
$|\varphi_{T}^\perp\>=(|0\>-e^{-iT}|1\>)/\sqrt{2}$.
That is, in the following, we employ the parametrization $P_{T,p}(\tau)$ by the pair $(T,p)$ for the probability of the measurement to yield outcome $\tau$. 
When the input state is $\rho_{T,p}$, 
it can be expressed as
 \begin{align}
P_{T,p}(\tau)&:=
\Tr \left[\rho_{T,p}M_{\tau}\right]\nonumber\\
&=\frac{1}{2\pi}(1+(2p-1) \cos (\tau-T)).
\end{align}
The Fisher information of $P_{T,p}(\tau)$ can be derived from a lengthy  but straightforward calculation, and we have
\begin{align}
\left(F_{{\rm loc},p}\right)_{TT}&=
\int_0^{2\pi} \left(\frac{\partial \log P_{T,p}(\tau)}{\partial T}\right)^2 
P_{T,p}(\tau) d\tau \nonumber \\
&=1-\sqrt{1-(2p-1)^2}\nonumber\\
&=1-\sqrt{1-e^{-2\gamma T}}.\label{QFI-nuisance-loc}
\end{align}
To evaluate the inverse of the Fisher information matrix, we also need to calculate the off-diagonal elements.
We define the logarithmic likelihood derivative $l_{p}(\tau)$ 
with respect to $p$ as
\begin{align*}
l_{p}(\tau)&:=\frac{\partial \log P_{T,p}(\tau)}{\partial p}\\
&= \frac{2 \cos (\tau-T)
}{1+(2p-1) \cos (\tau-T)}.
\end{align*}
Then the off-diagonal components $F_{T,p}$ and $ F_{p,T}$ are given as
\begin{align}
\left( F_{{\rm loc},p}\right)_{T,p}=\left( F_{{\rm loc},p}\right)_{p,T}:=\int_0^{2\pi} l_{p}(\tau) 
\frac{\partial \log P_{T,p}(\tau)}{\partial T}P_{T,p}(\tau) d\tau =0.
\end{align}
Since the off-diagonal components are zero, we have 
$$\left(   F_{\rm loc}^{-1}\right)_{TT}=\left(   F_{{\rm loc},p}^{-1}\right)_{TT}=1/\left(  F_{{\rm loc},p}\right)_{TT}.$$
The value of $\left(   F_{\rm loc}^{-1}\right)_{TT}$ is thus given by the reciprocal of the r.h.s. of Eq. (\ref{QFI-nuisance-loc}).

\end{document}